\newif\ifpersonal
\personaltrue
\personalfalse 

\newif\iflipics
\lipicstrue
\lipicsfalse

\iflipics
\documentclass[a4paper,UKenglish]{lipics-v2016}
\else
\documentclass[runningheads]{llncs}
\fi

\usepackage[T1]{fontenc}
\usepackage{lmodern}

\usepackage{amsmath,amssymb,amsfonts}
\usepackage{microtype}
\usepackage[boxed]{algorithm}
\usepackage[noend]{algorithmic}
\usepackage{epsfig}
\usepackage{color}
\usepackage{xspace}
\usepackage{authblk}
\usepackage{subcaption}
\iflipics
\else
\usepackage[T1]{fontenc}
\usepackage{graphicx}
\usepackage{todonotes}
\fi

\iflipics
\else
\fi

\DeclareMathOperator{\opt}{opt}
\DeclareMathOperator{\opta}{opt_{aug}}
\providecommand{\eqdef}{:=}

\ifpersonal

\fi

\newcommand{\etal}{{\em et al.\ }\xspace}
\iflipics
\renewcommand{\paragraph}{\subparagraph}
\fi
\begin{document}

\iflipics
\else
\title{A New Approach for Approximating Directed Rooted Networks}
\author{Sarel Cohen\inst{1}\orcidID{0000-0003-4578-1245} 
\and
Lior Kamma \inst{1}\orcidID{0000-0002-2380-7195} 
\and
Aikaterini Niklanovits\inst{2}\orcidID{0000-0002-4911-4493}}
\authorrunning{S. Cohen et al.}
%
\institute{The Academic College of Tel Aviv-Yaffo, Israel
\email{\{sarelco, liorkm\}@mta.ac.il}\\
\and
Hasso Plattner Institute, University of Potsdam, Germany\\
\email{Aikaterini.Niklanovits@hpi.de}}
%
%
\fi
\maketitle         
\begin{abstract}
We consider the {\em $k$-outconnected directed Steiner tree} problem ($k$-DST). Given a directed edge-weighted graph $G=(V,E,w)$, where $V = \{r\} \cup S \cup T$, and an integer $k$, the goal is to find a minimum cost subgraph of $G$ in which there are $k$ edge-disjoint $rt$-paths for every terminal $t \in T$.
%
%
The problem is known to be NP-Hard. Furthermore, 
the question on whether a polynomial time, subpolynomial approximation algorithm exists for $k$-DST was answered negatively by Grandoni \etal (2018), by proving an approximation hardness of $\Omega(|T|/\log |T|)$ under NP$\neq$ZPP.

Inspired by modern day applications, we focus on developing efficient algorithms for $k$-DST in graphs where terminals have out-degree $0$, and furthermore constitute the vast majority in the graph.
%
We provide the first approximation algorithm for $k$-DST on such graphs, in which the approximation ratio depends (primarily) on the size of $S$. We present a randomized algorithm that finds a solution of weight at most $O(k|S| \log|T|)$ times the optimal weight, and with high probability runs in polynomial time.
\end{abstract}



\section{Introduction}

{\em Network design problems} deal with finding, inside a large network $G$, a cheap sub-network satisfying certain criteria. 
This class of problems captures a wide variety of theoretical problems, among which are the minimum spanning tree problem, minimum Steiner tree (or forest) problem or minimum $k$-connected subgraph in both directed and undirected graphs (see, e.g. \cite{KM05,KN07,GK11}).
%
In addition to its theoretical importance, network design research is of great interest for developers of practical networks such as telecommunication networks.

Our focus in this paper is the {\em minimum $k$-outconnected directed Steiner tree} problem ($k$-DST). An instance of $k$-DST consists of a directed edge-weighted graph $G=(V,E,w)$, a designated node $r \in V$ called the {\em root}, a subset $T \subseteq V \setminus \{r\}$ of vertices called {\em terminals} and an integer $k$. Nodes in $V \setminus (T \cup \{r\})$ are referred to as {\em Steiner nodes}, and the set of Steiner nodes will be denoted by $S$. The goal is to find a minimum-weight subgraph of $G$, in which for every $t \in T$ there are at least $k$ edge-disjoint $rt$-paths.
%
From a theoretical perspective, $k$-DST is one of the most fundamental and intensively studied problems in computer science. The most extensively studied variation of the problem is the {\em minimum directed Steiner tree} (DST), in which $k=1$. This is, in part, due to the fact that many combinatorial problems are known to reduce to DST. Notable such examples among network design problems are the minimum node-weighted Steiner tree \cite{KR95} in both directed and undirected graphs and the group Steiner problem \cite{GKR00}. 

Moreover, $k$-DST has many applications in network design, circuit layouts, and phylogenetic tree reconstruction \cite{hwang1992steiner} and
specifically in designing {\em multicast} schemes \cite{OP05,Voss06}. When designing a multicast routing scheme, the goal is to distribute information from a single source to multiple destinations, while minimizing the usage of network resources (e.g. cost of links). To ensure that the network is robust against link failures, we require that the network contains several source-destinations routes that are, in a sense, mutually independent. The connection to $k$-DST is straightforward. Modeling the network as a directed graph, where edges model point-to-point links, a multicast routing scheme which can survive $k-1$ failures is a $k$-outconnected Steiner tree.

\paragraph{Networks Consisting of Mostly Terminals.} 
Upon approaching $k$-DST, one might be inclined to think of $T$ as a "small" designated subset of $V$. Indeed, in the majority of previous work, the approximation factors of approximation algorithms for $k$-DST mainly depend on the size of the terminal set $T$. In general, however, this need not be the case. In fact, modern day applications present networks in which the lion's share of the network consists of terminals. Consider, for 
example the problem of efficiently and reliably distributing critical healthcare information, including patient data and medical resources, ensuring minimal transmission costs. 
In particular, the root in this case is a central health database, the Steiner nodes are regional hospitals and clinics and the terminal set consists of remote clinics in rural areas, which are the majority of health sites and do not transmit any information.

%
From a theoretical point of view, $k$-DST is still an interesting problem when restricted to instances in which most of the nodes of $G$ are terminals and have no outgoing edges. Specifically, known hardness results for the problem hold even when $|V \setminus T| = O(|V|^{1/d})$ for some constant $d \in \mathbb{N}$. Moreover, there is no evidence that the problem is easier (even to approximate) when restricted to instances where $|V \setminus T|$ is polylogarithmic in $|V|$.

\paragraph{Related Work.}
$k$-DST is known to be NP-hard, as it models, for example, the Steiner tree in undirected graphs as a special case \cite{Karp72}.
However, many special cases are widely studied with positive results.
For example, the case where $T =V \setminus \{r\}$, known as \emph{minimum $k$-outconnected spanning subgraph problem} is proven to be polynomially time solvable \cite{Edmonds67}, as was also generalized by Frank and Tardos \cite{frank1988generalized} (see also \cite{frank2009rooted}).
Their algorithm was used as a subroutine to obtain a $2$-approximation algorithm for the undirected variant and its generalization \cite{khuller1994biconnectivity}, which is proven to be NP-hard.

When the existence of non-terminal vertices in the graph is also considered, the complexity of $k$-DST is increased, and the focus hence swifts to approximation algorithms.
For the case where $k=1$ for example, Charikar \etal \cite{CCC99} presented a family of greedy approximation algorithms ${\cal A}_d$, indexed by an integer $d>1$. By employing a technique introduced by Zelikovsky \cite{Zelikovsky97}, for every $d>1$ they showed a reduction of the instance $G$ to a new instance $G'$ consisting of a $d$-layered graph. The algorithm ${\cal A}_d$ returns a feasible solution (in $G$) of weight at most $O(d^3|T|^{1/d})$ times the optimal weight, and runs in time $|V|^{O(d)}$. A similar result has been given by Rothvo{\ss} \cite{Rothvoss11}, who applied the Lasserre hierarchy in order to obtain a family of algorithms giving similar approximation guarantees.
This result was recently extended to the case $k=2$ by Grandoni and Laekhanukit~\cite{GL16}. For every $d>1$, they develop a polynomial time approximation algorithm of ratio $\tilde{O}(d^3 |T|^{1/d})$ and running time $O(|V|^{O(d)})$.
For the undirected case of the problem when restricting the value of $k$ to be $1$, Byrka \etal provided a $1.39$-approximation algorithm \cite{byrka2013steiner}.

Moreover, the {\em directed Steiner tree} problem (i.e. $k=1$), is solved by a trivial algorithm in time $2^{|S|}\cdot |V|^{O(1)}$, thus if $S$ is ``small'' (i.e. $|S| \le O(\log |V|)$) the problem is solvable in polynomial time. On the other hand, the Dreyfus–Wagner algorithm \cite{DW71} solves the problem in time $O( 3^{|T|} \cdot |V|^2 )$, thus showing that it can also be solved in polynomial time if the number of terminals is ``small''. 
It is hence natural to consider developing algorithms for $k$-DST whose running time or approximation ratio is based on the amount of terminals.
Such approximation algorithms for $k$-DST are developed only for the special cases where $G$ is an $L$-layered graph, by Laekhanukit who gave a  $\mathcal O (k^L L  log |T|)$-approximation algorithm \cite{DBLP:conf/icalp/Laekhanukit16}, and for the case where $G$ is a quasi-bipartite graph for which Chan \etal gave a polynomial time $\mathcal O(\log |T|\log k)$-approximation algorithm \cite{DBLP:conf/approx/ChanLW020}.

When it comes to hardness of $k$-DST much more results have been obtained.
Cheriyan \etal \cite{CLNV14} proved it to be at least as hard to approximate as the label cover problem. Therefore there is no $2^{\log^{1- \varepsilon}|V|}$-approximation for $k$-DST for any fixed $\varepsilon > 0$ unless $NP \subseteq DTIME(n^{polylog(n)})$. Laekhanukit \cite{L14} showed that unless $NP=ZPP$, the $k$-DST problem admits no $k^{1/4 - \varepsilon}$-approximation for any $\varepsilon>0$.
We note that this hardness result holds even when restricted to instances in which $|V \setminus T| = O(|V|^{1/d})$ for some constant $d \in \mathbb{N}$ via a simple polynomial reduction.
Another polynomial reduction of $k$-DST that provides us with a lower bound for approximating it is the one to the {\em set cover} problem.
Nelson \cite{Nelson07} showed that for every $c < 0.5$, set cover cannot be approximated in polynomial time within $2^{\log^{1 - \delta_c(m)}(m)}$, unless SAT on $m$ variables can be solved in time $2^{O(2^{\log ^{1 - \delta_c(m)}(m)})}$, where $\delta_c(m) = 1 - \frac{1}{(\log \log m)^c}$.
Through a simple polynomial time reduction we are able to obtain a lower bound for approximating $k$-DST (in fact, even DST) in terms of $|S|$.
In particular we conclude that $k$-DST cannot be approximated in polynomial time to within $2^{\log^{1 - \delta_c(|S|)}(|S|)}$, unless SAT on $m$ variables can be solved in time $2^{O(2^{\log ^{1 - \delta_c(m)}(m)})}$. One of the most recent results is the one from Liao \etal \cite{DBLP:conf/icalp/LiaoCL022} where the improved the prreviously known lower approximation hardness to $\Omega(|T|/ \log |T|))$ under NP$\neq$ ZPP, $\Omega ( 2^{k/2}/k)$ under NP$\neq$ ZPP and $\Omega(k/L)^{L/4}$ on $L$-layered graphs for $L\leq O(\log n)$.

\subsection{Our Results}
Let $G=(\{r\} \cup S \cup T,E,w)$, $k \in \mathbb{N}$ be an instance of $k$-DST. Assume in addition that every edge $uv \in E$ satisfies that $u \notin T$. The main contribution of this paper is providing the first approximation algorithm for $k$-DST on graphs where the terminals do not transmit any information. Our algorithm is a randomized Las Vegas algorithm, which runs, with high probability, in polynomial time. 
\begin{theorem} \label{th:main}
There is a randomized $O(k|S| \log|T|)$-approximation algorithm for $k$-DST when the outdegree of the vertices in $T$ is $0$, that runs in polynomial time with high probability.
\end{theorem}
This approximation guarantee depends primarily on the size of $S$, and up to a logarithmic factor is independent of the size of $T$. If $|S| = O(|T|^{1/d})$ for some constant $d>1$, the approximation ratio of our algorithm is comparable with that of Charikar \etal \cite{CCC99}. Moreover, if $|S|$ is polylogarithmic in $|T|$, we achieve a polylogarithmic approximation ratio (in terms of $|V|$) in polynomial time.

\subsubsection{Overview and Techniques.}
Throughout this paper let $G=(\{r\}\cup S \cup T,E,w),k$ be an instance of $k$-DST, where $V=\{r\}\cup S \cup T$, and
${\cal U} \eqdef \{U \subseteq V \setminus \{r\} : U \cap T \ne \emptyset\}\;.$
An $e=xy \in E$ {\em covers } $U \in {\cal U}$ if $e$ enters $U$, i.e. $x \notin U$ and $y \in U$. Given a set $F \subseteq E$, denote by 
$\varrho_F(U) \eqdef \{xy \in F : x \notin U,\; y \in U\}$ the set of edges in $F$ covering $U$.
By the minimum-cut maximum-flow theorem, a set $F \subseteq E$ is a feasible solution for $k$-DST if and only if every $U \in {\cal U}$ is covered by at least $k$ edges, i.e. $|\varrho_F(U)| \ge k$. 

\paragraph{Connectivity Augmentation. } 
An intermediate problem between DST and $k$-DST is the one of increasing the rooted outconnectivity of a graph by one (connectivity augmentation problem). Formally, let $\ell \in \mathbb{N}$. Given a graph $G=(V,E,w)$, a root $r \in V$, a set of terminals $T \subseteq V \setminus \{r\}$, and a set $E_\ell \subseteq E$ of edges such that in the subgraph $(V,E_\ell)$ of $G$ there are $\ell$ edge-disjoint $rt$-paths for every $t \in T$, we seek a minimum cost set $F \subseteq E \setminus E_\ell$ such that in $(V,E_\ell \cup F)$ there are $\ell+1$ edge disjoint $rt$-paths for every $t \in T$. An additional result presented in this paper is the following.

\begin{theorem}\label{th:approxAugm}
There is a randomized approximation algorithm for the connectivity augmentation problem when the outdegree of the vertices in $T$ is $0$, that constructs an $O(|S|\log|T|)$-approximate solution. Also, with probability at least $1 - \frac{\log|T|}{2^{|E|}}$ the algorithm runs in polynomial time.
\end{theorem}
The connection to $k$-DST is quite clear. Starting with an empty set of edges, our approximation algorithm iteratively finds small-weight solutions to the connectivity augmentation problem $k$ times, and produces edge sets $E_1, \ldots, E_k \subseteq E$ whose union is a small-weight feasible solution to $k$-DST, incurring an additional factor of $k$ in the approximation ratio.
We note that very often when solving connectivity problems by repeatedly invoking a connectivity augmentation mechanism, it is possible to reduce the additional factor to the approximation ratio from $k$ to $\log k$, using a standard linear-programming scaling technique introduced by Goemans \etal \cite{DBLP:conf/soda/GoemansGPSTW94}. The algorithm presented in this paper for the augmentation problem, however, is invoked on an auxiliary graph, rather than on $G$. This auxiliary graph, and thus the induced linear program, change between invocations, and hence it is not clear whether the scaling technique is applicable.

Given an instance to the connectivity augmentation problem, and applying the minimum-cut maximum-flow theorem once more, we deduce that a set $F \subseteq E \setminus E_\ell$ is a feasible solution to the problem if and only if $|\varrho_{E_\ell \cup F}(U)| \ge \ell+1$ for every $U \in {\cal U}$. From our assumptions on $E_\ell$, we have $|\varrho_{E_\ell}(U)| \ge \ell$ for every $U \in {\cal U}$. In addition $F \cap E_\ell = \emptyset$, and therefore $F$ is a feasible solution if and only if $\varrho_F(U) \ne \emptyset$ for all $U \in {\cal U}_\ell$, where ${\cal U}_\ell \eqdef \{U \in {\cal U} : |\varrho_{E_\ell}(U)| = \ell\} \;.$ Hence, the augmentation problem is formulated as the following integer program.
\begin{equation}
\arraycolsep=1.4pt\def\arraystretch{1.6}
\begin{array}{lll}
min \sum_{e \in E}{w_ex_e}& \qquad s.t.\qquad &\sum_{e \in \varrho_{E \setminus E_\ell}(U)}{x_e} \ge 1 \quad \quad  \forall U \in {\cal U}_\ell,\\
 &   &x_e \in \{0,1\} \quad \quad \forall e \in E
\end{array}
\tag{IP}
\label{eq:IP}
\end{equation}
\paragraph{Implicit Hitting Set.} In the {\em hitting set} problem, we are given a ground set $E$, with weights $\{w_e\}_{e \in E}$, and a collection ${\cal S} \subseteq 2^E$ of subsets of $E$. The goal is to find a minimum weight subset $F \subseteq E$ such that $F \cap U \ne \emptyset$ for all $U \in {\cal S}$. The problem is known to be NP-hard \cite{Karp72}, and can be approximated by an $O(\log|{\cal S}|)$-factor via a greedy algorithm. Karp \etal \cite{KM13} observed that many combinatorial optimization problems possess the following property. Given an instance $\Pi$, we construct in polynomial time an instance $E,w,{\cal S}$ to the hitting set problem, where $E,w$ are given explicitly, and ${\cal S}$, whose size might be exponential in the size of $\Pi$, is given implicitly by a membership oracle. 
Karp \etal named this setting the {\em implicit hitting set} problem, and observed that the well-known greedy algorithm for the hitting-set problem cannot be used to approximate this type of problems, as it examines all sets in ${\cal S}$. 

We show that under reasonable assumptions we still achieve a randomized polynomial time $O(\log |{\cal S}|)$-approximation algorithm. Specifically, we slightly modify the well-known randomized rounding algorithm for hitting set, and devise a $O(\log |{\cal S}|)$-approximation algorithm that, with very high probability, say $\ge 1 - 2^{-|E|}$, runs in time polynomial to $|E|$. Details are deferred to Appendix~\ref{app:hitSetReg}.
The connectivity augmentation problem can be viewed as an implicit hitting set instance, where ${\cal S} \eqdef \{ \varrho_{E \setminus E_\ell}(U) : U \in {\cal U}_\ell \}$.
However, even for $\ell = 0$ we get $\log|{\cal S}| = \Omega(|V|)$. Achieving a $|T|$-approximation for the problem is immediate (increase the $rt$-connectivity by one for each $t \in T$ separately), hence we seek to exploit its structure, and specifically the structure of ${\cal U}_\ell$ to improve our guarantees.
%
%
%

A set-family ${\cal G} \subseteq 2^V$ is called {\em intersecting}, if when $X,Y \in {\cal G}$ and $X \cap Y \ne \emptyset$, then $X \cap Y, X \cup Y \in {\cal G}$. 
Frank \cite{Frank79} observed that if $T=V \setminus \{r\}$, the family ${\cal U}_\ell$ is intersecting, and also showed that the problem of covering intersecting set-families by directed edges can be solved optimally efficiently. In general, however, ${\cal U}_\ell$ is not intersecting. 
In a following paper \cite{Frank99} Frank coined the refined notion of $T$-Intersecting Families.


\begin{definition}
A set family ${\cal G} \subseteq 2^V$ 
is called {\em $T$-intersecting} if for every $X \in {\cal G}$, $X \cap T \ne \emptyset$; and for every $X,Y \in {\cal G}$, if $X \cap Y \cap T \ne \emptyset$ then $X \cap Y , X \cup Y \in {\cal G}$.
\end{definition}
The proof of Proposition \ref{p:tight} is moved to the Appendix due to lack of space.
\begin{proposition} \label{p:tight}
${\cal U}_\ell$ is $T$-intersecting.
\end{proposition}
Frank considered the problem of covering $T$-intersecting families, and showed that if every edge of $G$ enters $T$, then a $T$-intersecting family can be covered efficiently using the primal dual approach. For arbitrary graphs, however, the problem of optimally covering a $T$-intersecting family of sets is NP-hard as it models the directed Steiner tree problem as a special case.
By taking advantage of several structural properties of $T$-intersecting families, we present a mechanism for constructing a small-weight cover of ${\cal U}_\ell$. The mechanism is presented in detail in Section~\ref{sec:high}. Loosely speaking, we iteratively prune ${\cal U}_\ell$, covering small sub-families of ${\cal U}_\ell$ one at a time. Since the approximation ratio of the hitting set randomized algorithm depends on the size of the family to be covered, the problem of covering a small sub-family of ${\cal U}_\ell$ can be better approximated. We show that by pruning ${\cal U}_\ell$ carefully, the process ends after a small number of iterations, thus constructing a feasible solution with better approximation guarantees.

\paragraph{Strict Cores of $T$-Intersecting Families.} 
In \cite{Nutov09} Nutov studied rooted connectivity problems and defined the notion of a {\em core} of a set family. Given a set family ${\cal G} \subseteq 2^V$, a set $X \in  {\cal G}$ is a core of ${\cal G}$ if it contains exactly one inclusion-minimal set of ${\cal G}$. Nutov gave an algorithm for covering so-called bi-uncrossable set families by iteratively covering the sub-families of cores. Formally, the iterative process is as follows. Given a set family ${\cal G}$, find a small-weight edge set $F_1 \subseteq E$ that covers all cores of ${\cal G}$. Denote by ${\cal G}^{F_1}$ the family of all sets in ${\cal G}$ not covered by $F_1$. Continue iteratively by finding $F_2,\ldots,F_t$ such that for every $j \in [t]$, $F_j$ covers all cores of ${\cal G}^{F_1 \cup \ldots \cup F_{j-1}}$.
Nutov additionally bounded the number of iterations needed to cover bi-uncrossable set families, when iteratively covering the sub-families of cores.
In Section~\ref{sec:high} we show that by iteratively covering the sub-family of cores of the $T$-intersecting family ${\cal U}_\ell$, we can guarantee that no more than $\log|T|$ iterations are needed. 
However we also show, the sub-family of cores can be as large as the entire ${\cal U}_\ell$, and hence our approximation guarantees are still too large. 

Therefore, in this paper we present the new notion of {\em strict cores}. A strict core of ${\cal U}_\ell$ is a set $X \in {\cal U}_\ell$ that contains exactly one inclusion-minimal set $C \in {\cal U}_\ell$, and satisfies $X \cap T = C \cap T$. We show that the sub-family of strict cores of ${\cal U}_\ell$ is significantly smaller than the sub-family of cores, and give a polynomial time approximation algorithm for covering the sub-family of strict cores. As we show in Section~\ref{sec:covCor} we can make no guarantee, however, on the number of iterations needed to iteratively cover the sub-family of strict cores. The main technical crux of this paper is therefore the construction of a new auxiliary graph in which every edge cover to the family of strict cores also covers the family of cores.

\iflipics
\else
\section{Approximating $k$-Outconnected Steiner Tree}
\fi
We begin the proof of Theorem~\ref{th:main} by showing that it is implied by Theorem~\ref{th:approxAugm}.
To this end, let $G = (\{r\}\cup S\cup T,E,w)$ and $k$ be an instance of $k$-DST. Denote by $\opt(G)$ the value of an optimal solution, and let ${\cal U} = \{U \subseteq V \setminus \{r\} : U \cap T \ne \emptyset\}$.
We present an iterative algorithm that finds a feasible solution for $k$-DST, by iteratively increasing the connectivity between the root and the terminals. The algorithm performs $k$ iterations numbered $0,1,\ldots,k-1$. For every $0 \le \ell \le k-1$, in the beginning of the $\ell$th iteration, the algorithm holds a set $E_{\ell} \subseteq E$ of edges, starting with $E_0 = \emptyset$, such that in $(V,E_{\ell})$, there are at least $\ell$ edge-disjoint $rt$-paths for every $t \in T$. During the $\ell$th iteration, the algorithm invokes the $O(|S|\log|T|)$-approximation algorithm whose existence is implied by Theorem~\ref{th:approxAugm}, to find an approximate solution $F \subseteq E \setminus E_{\ell}$ to the connectivity augmentation instance given by $G,E_{\ell}$. 
The algorithm is described in detail as Algorithm~\ref{alg:iter}. 

\begin{algorithm}
\begin{algorithmic}[1]
\STATE $E_0 = \emptyset$.
\FOR{$\ell = 0,1,\ldots, k-1$}
\STATE let $F \subseteq E \setminus E_{\ell}$ be an $O(|S| \log|T|)$-approximate solution to the connectivity augmentation problem on $G,E_{\ell}$ (apply the algorithm in Theorem~\ref{th:approxAugm}). \label{l:connAug}
\STATE let $E_{\ell+1} = E_{\ell} \cup F$.
\ENDFOR
\RETURN $E_k$.
\end{algorithmic}
\caption{Approximation Algorithm for $k$-DST}
\label{alg:iter}
\end{algorithm}

Through a simple induction we are able to see that for every $t \in T$, there are $k$ edge-disjoint $rt$-paths in $(V,E_k)$.
This also implies that the edge set constructed by Algorithm~\ref{alg:iter} is a feasible solution to $k$-DST. The following lemma shows that the weight of this set is at most $O(k |S|\log|T|)$ times the optimum.
\begin{lemma}
$w(E_k) \le O(k |S|\log|T|) \cdot \opt(G)$.
\end{lemma}

\begin{proof}
Let $0 \le \ell \le k-1$. Consider an optimal solution $F^* \subseteq E$ to the $k$-DST problem on $G$. Then for every $U \in {\cal U}$, $|\varrho_{F*}(U) | \ge k$. Let $U \in {\cal U}$ be such that $|\varrho_{E_\ell}(U)| = \ell < k$. Then $|\varrho_{F^* \setminus E_\ell}(U)| \ge k - \ell \ge 1$.
Therefore $F^*\setminus E_\ell$ is a feasible solution for the augmentation problem defined by $G, E_{\ell}$. We conclude that the set $F$ constructed in line~\ref{l:connAug} satisfies $w(F) \le O(|S| \log |T|)\opt(G)$. The lemma follows.
\end{proof}

To conclude the proof of Theorem~\ref{th:main} it remains to show that the algorithm runs in polynomial time with high probability. Theorem~\ref{th:approxAugm} guarantees that line~\ref{l:connAug} always returns an approximate solution, and it runs in polynomial time with probability at least $1 - \tfrac{\log |T|}{2^{|E|}}$. Since Algorithm~\ref{alg:iter} performs $k$ iterations, we get by a union bound that Algorithm~\ref{alg:iter} runs in polynomial time with probability at least $1 - \tfrac{k\log |T|}{2^{|E|}}$, and since $k \log |T| \le |E|^2$, this probability is very high. 
The remainder of the paper is devoted to the proof of Theorem~\ref{th:approxAugm}.

\section{Increasing Rooted Connectivity By One}\label{sec:high}
In this section, we prove Theorem~\ref{th:approxAugm} by presenting an $O(|S|\log|T|)$-approximation algorithm for the connectivity augmentation problem.
Let $G = (\{r\}\cup S \cup T,E,w)$ be the instance graph, and $E_{\ell} \subseteq E$ be a set of edges such that in $(V,E_\ell)$ there are $\ell$ edge-disjoint $rt$-paths for every $t \in T$. Let 
${\cal F} \eqdef \{ U \subseteq V \setminus \{r\} : U \cap T \ne \emptyset \; and \; |\varrho_{E_\ell}(U)|=\ell \},$
then Proposition~\ref{p:tight} implies ${\cal F}$ is $T$-intersecting. By the min-cut max-flow theorem, a set $F \subseteq E \setminus E_{\ell}$ is a feasible solution for the augmentation problem if and only if $\varrho_F(U) \ne \emptyset$ for every $U \in {\cal F}$. The linear program below can be viewed as a fractional relaxation of the augmentation problem.
\begin{equation}
\arraycolsep=1.4pt\def\arraystretch{1.6}
\begin{array}{lll}
min &\sum_{e \in E}{w_ex_e}\qquad s.t. \qquad&\sum_{e \in \varrho_{E \setminus E_\ell}(U)}{x_e} \ge 1 \quad \quad  \forall U \in {\cal F},\\
& & x_e \ge 0 \quad \quad \forall e \in E
\end{array} \; .
\tag{LP}
\label{eq:LP}
\end{equation}
Fix some $x \in \mathbb{R}_+^E$. For every $e \in E$, we think of $x_e$ as the capacity of $e$ in $G$. Define $\hat{x} \in \mathbb{R}_+^E$ by $\hat{x}_e = 1$ for all $e \in E_\ell$ and $\hat{x}_e=x_e$ otherwise. Then $x$ is a feasible solution for \eqref{eq:LP} if and only if the capacity of a minimum $rt$-cut in $(G,\hat{x})$ is at least $\ell+1$ for every $t \in T$. Therefore, the feasibility of $x$ can be verified in time polynomial in the size of $G$. Moreover, if $x$ is not feasible, there is some $t \in T$ such that in $(G,\hat{x})$ there is an $rt$-cut of capacity less than $\ell+1$. Such a cut $U \subseteq V$ is a violated constraint and is found in polynomial time.
The program \eqref{eq:LP} can therefore be solved efficiently using a separation oracle. With high probability we thus efficiently obtain an $O(\log|{\cal F}|)$-approximate solution for the connectivity augmentation problem. However, in general $\log|{\cal F}| = \Omega(|V|)$, hence the approximation factor might be too large, since obtaining approximation of factor $|T|$ times the optimal value is trivial (increase the $rt$-connectivity for every $t \in T$ separately). Note that we did not use any properties of ${\cal F}$, other than finding an efficient separation oracle for \eqref{eq:LP}. Specifically, we did not use the structural properties of $T$-intersecting families. 

In what follows, we present a mechanism for covering the $T$-intersecting family ${\cal F}$ by edges.
Loosely speaking, we iteratively prune ${\cal F}$ by covering small sub-families. 
For each such sub-family we efficiently find a small weight cover. By choosing the sub-families carefully, the process ends after a small number of iterations. More formally, we look for sub-families ${\cal F}_1,\ldots,{\cal F}_t \subseteq {\cal F}$ such that (i) if $F \subseteq E \setminus E_\ell$ covers ${\cal F}_1,\ldots,{\cal F}_t$ then $F$ covers ${\cal F}$; (ii) for every $j \in [t-1]$, given a cover $F \subseteq E \setminus E_\ell$ for ${\cal F}_1,\ldots,{\cal F}_j$, we can efficiently find a small-weight cover $F' \subseteq E \setminus E_\ell$ for ${\cal F}_{j+1}$; and (iii) $t$ is small.

We start by exhibiting several fundamental properties of ${\cal F}$.
Given a set $F \subseteq E \setminus E_\ell$, a set $X \in {\cal F}$ is called {\em $F$-tight} if $F$ does not cover $X$, i.e. $\varrho_F(X) = \emptyset$. Denote by ${\cal F}^F$ the family of $F$-tight elements of ${\cal F}$. A crucial property of ${\cal F}$ is that the structural property of being $T$-intersecting is, in a sense, hereditary.
In particular it is easy to prove that ${\cal F}^F$ is $T$-intersecting as follows.
Let $X,Y \in {\cal F}^F$ be such that $X \cap Y \cap T \ne \emptyset$. Then $X \cap Y, X \cup Y \in {\cal F}$. Furthermore, by submodularity of the cut function,
$0 = |\varrho_F(X)| + |\varrho_F(Y)| \ge |\varrho_F(X \cup Y) | + |\varrho_F(X \cap Y)| \ge 0 \;.$ 
Therefore $|\varrho_F(X \cup Y) | = |\varrho_F(X \cap Y)| = 0$, and thus $X \cap Y$ and $X \cup Y$ are $F$-tight elements of ${\cal F}$, that is $X \cup Y, X \cap Y \in {\cal F}^F$.

The following definition, and the two lemmas that follow it are inspired by a similar observation by Nutov \cite{Nutov09} for problems regarding undirected graphs.
The proofs of Lemmas \ref{l:partition},\ref{l:half} are moved to the Appendix due to space restrictions.
\begin{definition}\label{def:core}
Let ${\cal M}({\cal F}^F)$ denote the family of inclusion-minimal elements of ${\cal F}^F$. A set $X \in {\cal F}^F$ is called an {\em ${\cal F}^F$-core} if it contains exactly one inclusion-minimal element of ${\cal F}^F$.
\end{definition}
\begin{lemma} \label{l:partition}
Let $X \in {\cal F}^F$. Then for every $C \in {\cal M}({\cal F}^F)$, either $C \cap X \cap T = \emptyset$ or $C \subseteq X$. In particular, for every distinct $C,D \in {\cal M}({\cal F}^F)$, $C \cap D \cap T  = \emptyset$.
\end{lemma}

\begin{lemma} \label{l:half}
If $F' \subseteq E \setminus E_\ell$ covers all ${\cal F}^F$-cores. Then $|{\cal M}({\cal F}^{F \cup F'})| \le \frac{1}{2}|{\cal M}({\cal F}^F)|$.
\end{lemma}

Lemma~\ref{l:half} implies that by iteratively covering the sub-family of tight cores, we cover ${\cal F}$ after at most $\log |{\cal M}({\cal F})|$ iterations. Since $|{\cal M}({\cal F})| \le |T|$, we conclude that the process terminates after at most $\log |T|$ iterations. 
The next section constitutes the technical crux in the proof of Theorem~\ref{th:approxAugm}.

\subsection{Covering Cores Cheaply and Efficiently}\label{sec:covCor}
In this section we conclude proof of Theorem~\ref{th:approxAugm} by giving an $O(|S|)$-approximation algorithm for the problem of covering the family of ${\cal F}^F$-cores. A na{\"i}ve approach suggested by the preceding discussion is to simply apply the implicit hitting set approximation algorithm to cover the family of ${\cal F}^F$-cores. 
However, the sub-family of ${\cal F}^F$-cores can be almost as large as the entire ${\cal F}^F$.
Hence the approximation factor might be as large as $\Omega(|V|)$. We therefore refine the definition of cores. This is where we diverge from \cite{Nutov09}, and from all previous work.

\begin{definition}
A set $X \in {\cal F}^F$ is called a {\em strict ${\cal F}^F$-core} if there exists $C \in {\cal M}({\cal F}^F)$ such that $X \cap T = C \cap T$, that is, $X$ and $C$ have the same set of terminals.
\end{definition}

One can easily verify, using Lemma~\ref{l:partition} that every strict ${\cal F}^F$-core contains exactly one inclusion-minimal set of ${\cal F}^F$. Thus every strict ${\cal F}^F$-core is also an ${\cal F}^F$-core.
%
%
As opposed to the family of tight cores, the family of strict cores is significantly small.
\begin{lemma}\label{l:approxValue}
The number of strict ${\cal F}^F$-cores is at most $2^{|S|} \cdot \left|T\right|$.
\end{lemma}
\begin{proof}
Let $X$ be a strict ${\cal F}^F$-core. Then there is a unique $C \in {\cal M}({\cal F}^F)$ such that $C \subseteq X$ and $C \cap T = X \cap T$. Therefore $X \setminus C \subseteq S$, and
hence the family of strict ${\cal F}^F$-cores is contained in $\{C \cup X' : C \in {\cal M}({\cal F}^F),\; X' \in 2^S \}$.
\end{proof}
By Lemma~\ref{l:approxValue}, if we can model the problem of covering all strict ${\cal F}^F$-cores as an implicit hitting set problem, the approximation factor of the randomized rounding algorithm is reduced to $O(|S|+\log|T|)=O(|S|)$.
However,
if $F' \subseteq E \setminus E_\ell$ covers all strict cores, it does not necessarily cover all cores. Nevertheless, not all is lost. 

In a sense, the essence of the inadequacy of covering only strict cores lies in the following scenario. 
Consider an edge $e$ leaving a covered core $X$ and entering a non-covered strict core $C$ covers the strict core, but does not reduce the number of non-covered minimal cores. Covering all strict cores with such edges attains no advancement to the algorithm. We thus turn to construct an auxiliary graph, in which this scenario cannot occur.


More formally, given $G, E_\ell$ and $F$, let $\opta(G, E_\ell)$ denote the weight of an optimal solution for the augmentation problem. We show that we can construct a new graph $G^F$ on the same vertex set $V$, in which (i) we can find in $G^F$ an edge set $A \subseteq E(G^F)$ covering all ${\cal F}^F$-cores such that $w(A) \le O(|S|)\opta(G,E_\ell)$; and (ii) given a cover $A \subseteq E(G^F)$ for the family of ${\cal F}^F$-cores, we can construct a cover $F' \subseteq E \setminus E_\ell$ of the ${\cal F}^F$-cores in $G$ such that $w(F')\le w(A)$. Moreover, the construction of $G^F$, as well as (i) and (ii), can be done in polynomial time with high probability. Note that unlike simple reductions commonly used in dealing with the directed Steiner tree problem (e.g. reducing the instance to an acyclic graph or a layered graph), the construction we present preserves the vertex set intact, and specifically does not change $|S|$. 

%

Following the discussion above, we want the weights to satisfy the triangle inequality, and therefore we first define a special form of metric completion of $G$. In classical literature, for every $u,v \in V$, the weight of the edge $uv$ in the metric completion of a weighted graph $G$ is defined to be the length of a shortest $uv$-path in $G$. By the minimum-cut maximum-flow theorem, for every $u,v \in V$, a shortest $uv$-path in $G$ is a minimum-weight edge-cover of the set $\{U \subseteq V \setminus \{u\} : u \in U\}$. Therefore an analogous way of viewing the definition of a metric completion is the following. For every $u,v \in V$, the weight of the edge $uv$ in the metric completion of $G$ is defined to be the minimum weight of an edge cover for the set $\{U \subseteq V \setminus \{u\} : u \in U\}$. In our setting not all such subsets $U \subseteq V \setminus \{u\}$ need to be covered. We therefore wish to refine the classical notion of a metric completion of $G$. Specifically, if every $X \in {\cal F}^F$ such that $v \in X$ satisfies $u \in X$, then in a sense we do not need to cover any $uv$-cut, since all {\em relevant} cuts (i.e. those belonging to ${\cal F}^F$) have been covered. We can therefore set the weight of the edge $uv$ to zero.

Formally we define a graph $G^0 = (V, (E \setminus E_\ell) \cup E^0, w^0)$ as $E^0 = \{uv: \text{for every $X \in {\cal F}^F$, if $v \in X$ then $u \in X$}\} \;$ and constructed through the following procedure.
For every $e \in (E \setminus E_\ell) \cup E^0$, we define $w^0_e=0$ if $e \in F \cup E^0$, and $w^0_e = w_e$ otherwise.
Next, let $G^1$ be the standard metric completion of $G^0$. That is, $G^1$ is the complete directed graph on $V$, where $w^1_{uv}$ is the length of a shortest $uv$-path in $G^0$. The weight assignment $w^1$ satisfies the triangle inequality, and is therefore an asymmetric metric.

Next, following the discussion above, we construct $G^F$ by removing from $G^1$ edges, which we can assert will not belong to any inclusion-minimal solution. Formally, $G^F$ is the subgraph of $G^1$ constructed as follows. For every $s \in \bigcup_{C \in {\cal M}({\cal F}^F)}{(C \cap T)}$ and for every $u \in V \setminus \left(\bigcup_{C \in {\cal M}({\cal F}^F)}{(C \cap T)} \right)$, if every $X \in {\cal F}^F$ that contains $s$ also contains $u$, remove from $G^1$ all edges outgoing from $u$. The proof of the claim that $G^F$ can be constructed in polynomial time given $G,E_\ell, F$ is deferred to Appendix~\ref{a:proofs}.
%

\begin{lemma} \label{l:reduce}
For every $A \subseteq E(G^F)$, if $A$ covers all ${\cal F}^F$-cores, we can construct in polynomial time an edge set $F' \subseteq E \setminus E_\ell$ that covers all ${\cal F}^F$-cores such that $w(F') \le w^1(A)$.
\end{lemma}
\begin{proof}
Let $A \subseteq E(G^F)$ be an ${\cal F}^F$-cores cover in $G^F$, $e = uv \in A$, and $P_e \subseteq (E \setminus E_\ell) \cup E^0$ be a shortest $uv$-path in $G^0$. Then $w^1_e = w^0(P_e) = \sum_{e' \in P_e}{w^0_{e'}} =\sum_{e' \in P_e \cap (E \setminus (E_\ell \cup F))}{w_{e'}}$.
Let $X \in {\cal F}^F$ be an ${\cal F}^F$-core covered by $e$. Then $u \notin X$ and $v \in X$. Hence there is an edge $e' = u'v' \in P_e$ that covers $X$. Since $X \in {\cal F}^F$ and $v' \in X$ and $u' \notin X$, it follows that $e' = u'v' \notin F \cup E^0$, and hence $e' \in E \setminus (E_\ell \cup F)$.
We conclude that $F' \eqdef \bigcup_{e \in A}{(P_e \cap (E \setminus (E_\ell \cup F))}$ covers all ${\cal F}^F$-cores, and
$w(F') = \sum_{e \in F'}{w_e} \le \sum_{e \in A}{\sum_{e' \in P_e \cap (E \setminus (E_\ell \cup F))}{w_{e'}}} = \sum_{e \in A}{w^1_e}= w^1(A).$
\end{proof}
Due to Lemma~\ref{l:reduce}, in order to find a low-cost cover $F' \subseteq E \setminus E_\ell$ for all ${\cal F}^F$-cores it suffices to find a low-cost cover $A \subseteq E(G^F)$. The following two lemmas show that such a cover exists, and that it is enough for $A$ to cover the strict ${\cal F}^F$-cores. 

\begin{lemma} \label{l:strictPossible}
There exists an edge set $A \subseteq E(G^F)$ such that $A$ covers all ${\cal F}^F$-cores and $w^1(A) \le \opta(G,E_\ell)$.
\end{lemma}

\begin{proof}
For every set $E' \subseteq E \setminus E_\ell$, from our construction we get that $w^1(E') \le w(E')$, and therefore $\opta(G^1,E_\ell) \le \opta(G,E_\ell)$. Hence, there is an edge set $A' \subseteq E(G^1)$ such that in the subgraph $(V,E_\ell \cup A')$ there are $\ell+1$ edge-disjoint $rt$-paths for all $t \in T$, and such that $w^1(A') \le \opta(G,E_\ell)$. We show that there is a set $A \subseteq E(G^F)$ such that in the subgraph $(V,E_\ell \cup A)$ there are $\ell+1$ edge-disjoint $rt$-paths for all $t \in T$ and $w^1(A) \le w^1(A')$. Clearly, such $A$ covers all strict ${\cal F}^F$-cores in $G^F$ and $w^1(A) \le \opta(G,E_\ell)$. We continue by induction on $|A' \setminus E(G^F)|$. If $|A' \setminus E(G^F)|=0$, take $A = A'$, and the result follows.
Otherwise, let $uv \in A' \setminus E(G^F)$. Then $u \in V \setminus \left(\bigcup_{C \in {\cal M}({\cal F}^F)}{(C \cap T)} \right)$ and there is $s \in \bigcup_{C \in {\cal M}({\cal F}^F)}{(C \cap T)}$ such that every $Z \in {\cal F}^F$ containing $u$ also contains $s$. By the definition of $A'$, every set $Z \in {\cal F}^F$ containing $s$ is covered by $A'$. 
Let $A'' = (A' \setminus \{uv\}) \bigcup \{sv\}$. By triangle inequality, $w^1_{sv} \le w^1_{su} + w^1_{uv}$. Since $su \in E^0$, $w^1_{su}=0$ we get that $w^1_{sv} \le w^1_{uv}$. Hence, $w^1(A'')\le w^1(A')$. 
Since $s \in \bigcup_{C \in {\cal M}({\cal F}^F)}{(C \cap T)}$, then $sv \in E(G^F)$, we get $|A'' \setminus E(G^F)| = |A' \setminus E(G^F)|-1$. 
\end{proof}

\begin{lemma} \label{l:strictEnough}
For every $A \subseteq E(G^F)$, $A$ covers all ${\cal F}^F$-cores if and only if $A$ covers all strict ${\cal F}^F$-cores.
\end{lemma}

\begin{proof}
Let $A \subseteq E(G^F)$. Clearly, if $A$ covers all ${\cal F}^F$-cores, then it covers all strict ${\cal F}^F$-cores. Assume, therefore, that $A$ covers all strict ${\cal F}^F$-cores. Let $X \in {\cal F}^F$ be an ${\cal F}^F$-core that is not a strict core. There is a unique $C \in {\cal M}({\cal F}^F)$ such that $C \subseteq X$. Let $X' = X \setminus (T \setminus C)$.
%
First note that since for every $u \in T \setminus C$, no edge leaves $u$, then by removing from $X$ all terminals in $T \setminus C$, no new edges enter $X'$. Therefore $X' \in {\cal F}$.
To see that $X' \in {\cal F}^F$, let $xy \in F$ be some edge, and assume $y \in X' \subseteq X$. 
Assume $y \in C$. Since $C \in {\cal F}^F$, then $x \in C \subseteq X'$. Otherwise, since $x \notin T$, it follows that $x \notin X \setminus X'$. Since $X \in {\cal F}^F$, it follows that $x \in X$, and thus $x \in X'$. Therefore $xy$ does not cover $X'$, and thus $X' \in {\cal F}^F$.
Next we note that $X' \cap T = C \cap T$, and thus $X'$ is a strict core.

By assumption on $A$, there is an edge $uv \in A$ such that $v \in X' \subseteq X$ and $u \notin X'$. If $u \in X$, then $u \in T$. Therefore $u \notin X$, and thus $X$ is covered by $A$.
\end{proof}
The following lemma shows that in $G^F$ we can find with high probability an approximate cover for the set of ${\cal F}^F$-cores. 

\begin{lemma} \label{l:strictAlgorithm}
There is a randomized algorithm that finds a set $A \subseteq E(G^F)$ such that $A$ covers all ${\cal F}^F$-cores and $w^1(A)\le O(|S|)\opta(G,E_\ell)$. Moreover, with probability at least $1 - 2^{-|E|}$ the algorithm runs in polynomial time.
\end{lemma}
The proof of Lemma~\ref{l:strictAlgorithm} is technically involved, however the main idea is quite straightforward. 
First, note that by Lemma~\ref{l:strictEnough}, it suffices to show an algorithm that finds an approximate cover for the set of strict ${\cal F}^F$-cores.
We show that a linear program relaxation of the latter has a polynomial time separation oracle, hence the randomized rounding algorithm for the hitting set finds such a cover in polynomial time with high probability.
The proof is given in detail in Appendix~\ref{a:proofs}.

We are now ready to prove Theorem~\ref{th:approxAugm}, by describing an $O(|S|\log|T|)$-approximation algorithm for the augmentation problem. The algorithm iteratively prunes ${\cal F}$ by covering the set of tight cores in every iteration. 
Starting with $F = \emptyset$, and during every iteration until $F$ covers ${\cal F}$, we add to $F$ a set $F' \subseteq E \setminus E_\ell$ of weight at most $O(|S|) \opta(G, E_\ell)$  that covers the family of ${\cal F}^F$-cores. The mechanism is presented in detail as Algorithm~\ref{alg:gen}. 

\begin{algorithm}
\begin{algorithmic}[1]
\STATE $j\gets 1$
\STATE $F_j \leftarrow \emptyset$
\WHILE{${\cal F}^F \neq \emptyset$}
\STATE {\em (implicitly)} let ${\cal F}_j$ be the set of ${\cal F}^F$-cores.
\STATE let $F' \subseteq E \setminus E_\ell$ be an edge set covering ${\cal F}_j$ such that $w(F') \le O(|S|) \opta(G, E_\ell)$. \label{line:aug}
\STATE $F \leftarrow F \cup F'$.
\STATE $j \leftarrow j+1$.
\ENDWHILE
\RETURN $F$.
\end{algorithmic}
\caption{$O(|S| \log |T|)$-Approximation Algorithm for the Connectivity Augmentation Problem}
\label{alg:gen}
\end{algorithm}
\begin{proof}[of Theorem~\ref{th:approxAugm}]
We first show that each iteration is executed in polynomial time with probability at least $1 - 2^{-|E|}$. 
Given a set $F \subseteq E \setminus E_\ell$, we efficiently (deterministically) check whether ${\cal F}^F = \emptyset$, by verifying that for every $t \in T$, there are $\ell+1$ edge-disjoint $rt$-paths in $(V,E_\ell \cup F)$.
To see that line~\ref{line:aug} runs in polynomial time with high probability, first note that given $G,E_\ell,F$, we can construct $G^F$ in polynomial time. 
%
By Lemma~\ref{l:strictAlgorithm}, there is a randomized algorithm that finds a set $A \subseteq E(G^F)$ such that $A$ covers all ${\cal F}^F$-cores and $w^1(A)\le O(|S|)\opta(G,E_\ell)$. Moreover, with probability at least $1 - 2^{-|E|}$ the algorithm runs in polynomial time. Lemma~\ref{l:reduce} then guarantees that we can construct from $A$ an edge set $F' \subseteq E \setminus E_\ell$ that covers all ${\cal F}^F$ cores, and such that $w(F') \le w^1(A) \le O(|S|)\opta(G,E_\ell)$.

Therefore each iteration can be done in polynomial time with probability at least $1 - 2^{-|E|}$. Lemma~\ref{l:half} ensures that Algorithm~\ref{alg:gen} terminates after at most $\log |T|$ iterations. Applying a union bound, we get that Algorithm~\ref{alg:gen} is an $O(|S| \log |T|)$-approximation algorithm for the connectivity augmentation problem, that runs in polynomial time with probability at least $1 - (2^{-|E|}\log |T|)$.

\end{proof}
\iflipics
\else
\bibliographystyle{splncs04}
\fi
\bibliography{DST}
\newpage
\appendix
%

\section{Implicit Hitting Set Randomized Rounding Algorithm}\label{app:hitSetReg}
Let $E,w,{\cal S}$ be an instance for the implicit hitting set problem. In this setting of implicit hitting set, the size of the input for the problem is polynomial in $|E|$, while ${\cal S} \subseteq 2^E$ can be much larger and is thus given implicitly. The fractional relaxation of the problem can be formulated as the following linear program.
\begin{equation}
\arraycolsep=1.4pt\def\arraystretch{1.6}
\begin{array}{lll}
min \sum_{e \in E}{w_ex_e}\quad
s.t.\quad &\sum_{e \in U}{x_e} \ge 1 &\quad \quad  \forall U \in {\cal S}\\
& x_e \ge 0 &\quad \quad \forall e \in E
\end{array}
\label{appEq:HSLP}
\end{equation}
A classical result shows that if $x^*$ is an optimal solution to the linear program \eqref{appEq:HSLP}, there is a polynomial time randomized rounding algorithm that returns with probability at least $\frac{1}{2}$ a feasible solution $F$ for the hitting set problem satisfying $w(F) \le  O(\log|{\cal S}|)w(x^*)$, where $w(x^*) = \sum_{e \in E}{w_ex^*_e}$. We can further claim that if there is a polynomial time separation oracle to the linear program \eqref{appEq:HSLP}, then we can efficiently verify whether $F$ is a feasible solution. 
This result is well known, however for sake of completeness we include the algorithm (presented as Algorithm~\ref{appAlg:HSRR}) and formal statement (Lemma~\ref{appC:hitSetEffApprox}). For a more detailed analysis see, e.g. \cite[Chapter~14]{Vazirani01}.
\begin{algorithm}[th]
\begin{algorithmic}[1]
\STATE let $x^* = \{x^*_e\}_{e \in E}$ be an optimal solution for \eqref{appEq:HSLP}.
\STATE $F \leftarrow \emptyset$ \label{l:empty}
\STATE repeat $O(\log|{\cal S}|)$ times independently.
\STATE \hspace{1pc} For every $e \in E$ independently, add $e$ to $F$ with probability $x^*_e$.
\IF{$F$ is a feasible solution to the hitting set problem and $w(F) \le O(\log|{\cal S}|)w(x^*)$}
\RETURN $F$.
\ELSE
\STATE go to line \ref{l:empty}.
\ENDIF
\end{algorithmic}
\caption{Hitting Set Randomized Rounding}
\label{appAlg:HSRR}
\end{algorithm}

\begin{lemma}\label{appC:hitSetEffApprox}
Algorithm~\ref{appAlg:HSRR} always returns a feasible solution $F$ of weight at most $O(\log|{\cal S}|)$ times the optimal weight. Moreover, if there is a polynomial time separation oracle to the linear program \eqref{appEq:HSLP}, then with probability at least $1 - 2^{-|E|}$ Algorithm~\ref{appAlg:HSRR} runs in polynomial time.
\end{lemma}

\section{Proof of Proposition \ref{p:tight}}
\begin{proof}
Let $X, Y \in {\cal U}_\ell$, and assume that $X \cap Y \cap T \ne \emptyset$. Therefore $X \cup Y, X \cap Y \in {\cal U}$ and furthermore, by submodularity of the cut function,
$$2 \ell = |\varrho_{E_\ell}(X)| + |\varrho_{E_\ell}(Y)| \ge |\varrho_{E_\ell}(X \cup Y) | + |\varrho_{E_\ell}(X \cap Y)| \;.$$ 
By the definition of $E_\ell$, $|\varrho_{E_\ell}(X \cup Y) | = |\varrho_{E_\ell}(X \cap Y)| = \ell$, and therefore  $X \cap Y,X \cup Y \in {\cal U}_\ell$.
\end{proof}

\section{Proofs for Section~\ref{sec:covCor}} \label{a:proofs}
In this section we complete the technical proofs of Section ~\ref{sec:covCor}.

\subsection{Proof of Lemma \ref{l:partition}}
\begin{proof}
Assume $C \cap X \cap T \ne \emptyset$, then by the previous claim, $C \cap X \in {\cal F}^F$. Since $C \in {\cal F}^F$ is inclusion-minimal, $C \cap X = C$ and thus $C \subseteq X$.
\end{proof}

\subsection{Proof of Lemma \ref{l:half}}

\begin{proof}
Let $C \in {\cal M}({\cal F}^{F \cup F'})$, then $C \in {\cal F}$ is not covered by $F \cup F'$, i.e.
$\varrho_{F \cup F'}(C)= \emptyset$. It follows that $\varrho_F(C) = \emptyset$, and therefore $C \in {\cal F}^F$. In addition, $\varrho_{F'}(C) = \emptyset$ and thus $C$ is not covered by $F'$. Since $F'$ covers all ${\cal F}^F$-cores, $C$ is not an ${\cal F}^F$-core. Hence, there are $X,Y \in {\cal M}({\cal F}^F)$ such that $X,Y \subseteq C$.
Denote ${\cal X}_C \eqdef \{X \in {\cal M}({\cal F}^F) : X \subseteq C\}$. Then ${\cal X}_C \subseteq {\cal M}({\cal F}^F)$, and we have shown that $|{\cal X}_C| \ge 2$.
Since ${\cal F}^{F \cup F'}$ is also a $T$-intersecting family, inclusion-minimal elements of ${\cal F}^{F \cup F'}$ are terminal-disjoint. Therefore $C$ is the unique element of ${\cal M}({\cal F}^{F \cup F'})$ containing $X,Y$.
It follows that $\{{\cal X}_C\}_{C \in {\cal M}({\cal F}^{F \cup F'})}$ is a partition of ${\cal M}({\cal F}^F)$, in which every set contains at least two elements. Therefore $|{\cal M}({\cal F}^{F \cup F'})| \le \frac{1}{2}|{\cal M}({\cal F}^F)|$.
\end{proof}

\subsection{$G^F$ can be constructed in polynomial time}
In order to prove that $G^F$ can be constructed efficiently, we need to show that we can efficiently decide for every $x \in V$ whether $x \in \bigcup_{C \in {\cal M}({\cal F}^F)}{C}$. 
\begin{claim}
Let $t \in T$. If there exists a tight set $X \in {\cal F}^F$ such that $t \in X$, then there exists a unique inclusion-minimal such set.
\end{claim}
\begin{proof}
Denote $R_t = \{X \in {\cal F}^F : t \in X\}$. 
Since $t \in T$, and since ${\cal F}^F$ is $T$-intersecting, we get that $\bigcap\limits_{X \in R_t}{X} \in {\cal F}^F$. Clearly it is an inclusion-minimal such set, and we similarly show it is unique.
\end{proof}
For every $t \in T$, denote by $C_t$ the unique inclusion-minimal tight set containing $t$. 
The following claim is straightforward, and it characterizes these minimal sets in a manner that allows to recover them efficiently given $G,E_\ell,F$.
\begin{claim}\label{c:minCuts}
For every $t \in T$, there exists a tight set $X \in {\cal F}^F$ containing $t$ if and only if the minimum $rt$-cut in $(V,E_\ell \cup F)$ contains $\ell$ edges. Moreover, $C_t$ is an inclusion minimal $rt$-cut of capacity $\ell$ and can therefore be computed in polynomial time.
\end{claim}
%
Intuitively, one may think that for every $t \in T$, $C_t$ is an inclusion-minimal element of ${\cal F}^F$. However, it turns out that this needs not be the case.
Therefore in general we cannot argue that $\{C_t\}_{t \in T} \subseteq {\cal M}({\cal F}^F)$. The following claim shows that the converse containment does hold.
\begin{proposition} \label{prop:minSets}
${\cal M}({\cal F}^F) \subseteq \{C_t\}_{t \in T}$. Furthermore, for every $t \in T$, $C_t \in {\cal M}({\cal F})$ if and only if there exists a strict core containing $t$.
\end{proposition}
\begin{proof}
Let $C \in {\cal M}({\cal F}^F)$, and let $t \in C \cap T$. Then $C$ is an inclusion-minimal $F$-tight set containing $t$. Therefore $C=C_t \in \{C_s\}_{s \in T}$.
Next, fix some $t \in T$. Clearly, if $C_t \in {\cal M}({\cal F})$, then $C_t$ is a strict core containing $t$. Assume that there exists a strict core $X$ containing $t$. Then there exists $C \in {\cal M}({\cal F}^F)$ such that $C \subseteq X$ and $X \cap T = C \cap T$. Since $t \in C_t \cap X \cap T = C_t \cap C \cap T$, and by minimality of both $C$ and $C_t$, and since $C_t \cap C$ is tight, it follows that $C_t\cap C = C_t=C \in {\cal M}({\cal F}^F)$.
\end{proof}
We are now ready to prove the main claim of this section.
\begin{proof}
We first show that $E^0$ (and thus $w^0$) can be constructed in polynomial time given $G,E_\ell,F$. Following the characterization of $\{C_t\}_{t \in T}$ and ${\cal M}({\cal F}^F)$ in Claim~\ref{c:minCuts} and Proposition~\ref{prop:minSets}, we can construct the family ${\cal M}({\cal F}^F)$ in polynomial time. Let $u,v \in V$, and consider the following algorithm for constructing $E^0$. 
\begin{algorithm}[H]
\begin{algorithmic}[1]
\FORALL{$C \in {\cal M}({\cal F}^F)$}
\STATE find in $(V, E_\ell \cup F)$ a minimum cut $U_C$ that separates $C \cup \{v\}$ and $\{r,u\}$.
\IF{$\min_{C \in {\cal M}({\cal F}^F)}|\varrho_{E_\ell \cup F}(U_C)| \ge \ell +1$}
\STATE add $uv$ to $E^0$.
\ENDIF
\ENDFOR
\end{algorithmic}
\label{alg:GFConst}
\end{algorithm}

Note first that the algorithm runs in polynomial time. We will show that the algorithm adds $uv$ to $E^0$ if and only if for every $X \in {\cal F}^F$, if $v \in X$ then $u \in X$. Assume first that there exists $X \in {\cal F}^F$ such that $v \in X$ and $u \notin X$. Then there exists some $C \in {\cal M}({\cal F}^F)$ such that $\{v\} \cup C \subseteq X \subseteq V \setminus \{r,u\}$. Therefore $X$ is a cut separating $C \cup \{v\}$, and since $X \in {\cal F}^F$, then the capacity of $X$ is $\ell$. Therefore $\varrho_{E_\ell \cup F}(U_C) \le \ell$ and $uv$ is not added to $E^0$.  Otherwise, for every $X \in {\cal F}^F$, if $v \in X$ then $u \in X$. Let $C \in {\cal M}({\cal F}^F)$, and let $U_C$ be a minimum cut that separates $C \cup \{v\}$ and $\{r,u\}$ in $(V, E_\ell \cup F)$. Then $C \subseteq U_C \subseteq V \setminus \{r\}$, and therefore $U_C \in {\cal U}$. Since $v \in U_C$ and $u \notin U_C$, then $U_C \notin {\cal F}^F$. Therefore $\varrho_{E_\ell \cup F}(U_C) \ge \ell +1$. Hence $uv$ is added to $E^0$ by the algorithm.

Given $G^0$, the metric completion $G^1$ can be constructed in polynomial time. 
Finally, for every $u \in V$ and $s \in T$, by Proposition~\ref{prop:minSets} we can decide in polynomial time whether $u \notin \bigcup_{C \in {\cal M}({\cal F}^F)}{C}$ and whether $s \in \bigcup_{C \in {\cal M}({\cal F}^F)}{C}$. In a similar manner to constructing $E^0$, we can decide efficiently whether to remove all edges leaving $u$. Therefore the set $E(G^F)$ can be constructed in polynomial time.
\end{proof}

\subsection{Proof of Lemma~\ref{l:strictAlgorithm}}



Consider the instance $E(G^F),w,{\cal S}^F$ for the implicit hitting set problem, where $${\cal S}^F= \{\varrho_{E(G^F)}(X) : \text{$X$ is a strict ${\cal F}^F$-core}\} \;,$$ and the corresponding fractional relaxation.
\begin{equation}
\arraycolsep=1.4pt\def\arraystretch{1.6}
\begin{array}{lll}
min \sum_{e \in E(G^F)}{w^1_ex_e}\quad
s.t. \quad&\sum_{e \in \varrho_{E(G^F)}(X)}{x_e} \ge 1 &\quad \quad  \text{for all strict ${\cal F}^F$-cores $X$}\\
& x_e \ge 0 &\quad \quad \forall e \in E(G^F)
\end{array}
\label{eq:LPCores}
\end{equation}
%

We present a polynomial time separation oracle for the linear program~\eqref{eq:LPCores}. That is, we show an algorithm which, given a vector $x \in {\mathbb R}_+^{E(G^F)}$ either reports that $x$ is feasible for \eqref{eq:LPCores} or returns a constraint violated by $x$. For every $t \in T$, the algorithm verifies that $x$ satisfies the constraints regarding strict cores containing $t$. Given some $t \in T$, we first check whether there is a strict core containing $t$. If there is no such strict core, we are done. Otherwise, we find the (unique) element $C \in {\cal M}({\cal F}^F)$ containing $t$, and construct an edge-capacitated auxiliary graph $H_t=(V,E',\{x^t_e\}_{e \in E'})$, where $E' = E(G^F) \cup E_\ell \cup F$ (note that $E'$ may contain parallel edges), and $x^t$ is defined as follows. For every $s \in T \setminus C$ we set $x^t_{rs}$ to $1$, as demonstrated in Figure~\ref{fig:auxiliary}. In addition, we define $x^t_e = 1$ for all $e \in E_\ell \cup F$. For all other edges $e$, we set $x^t_e = x_e$. We now find a minimum $rt$-cut in $H_t$. The key observation is that if the capacity of the cut is strictly less than $\ell+1$, then $V \setminus U$ is a strict ${\cal F}^F$-core that is violated by $x$ (see Figure~\ref{fig:auxiliary-c}). The detailed algorithm is given as Algorithm~\ref{alg:LPSep}.

\begin{figure}[ht]
  \begin{center}
  \begin{subfigure}[t]{.25\linewidth}
        \includegraphics[scale=0.5]{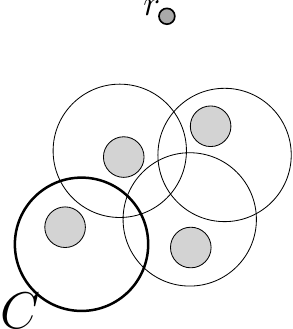}
        \caption{Given $t \in T$, the highlighted area represents the unique $C \in {\cal M}({\cal F}^F)$ such that $t \in C$. The white circles represent ${\cal M}({\cal F}^F) \setminus \{C\}$, while their intersection with $T$ is dark.}
        \label{fig:auxiliary-a}
   \end{subfigure} \hspace{2pc}
   \begin{subfigure}[t]{.25\linewidth}
        \includegraphics[scale=0.5]{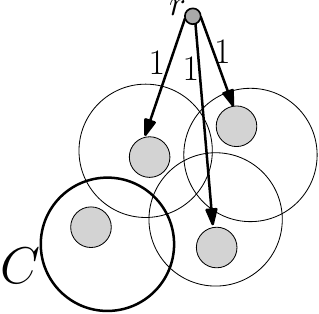}
        \caption{We add $1$-capacity edges from $r$ to $s$ for all $s \in T \setminus C$.}
        \label{fig:auxiliary-b}
   \end{subfigure} \hspace{2pc}
   \begin{subfigure}[t]{.25\linewidth}
        \includegraphics[scale=0.5]{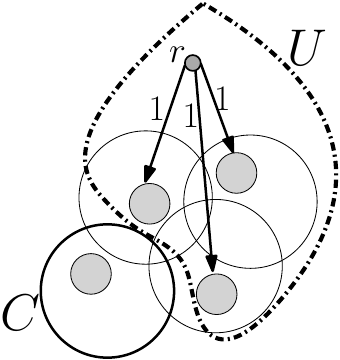}
        \caption{An $rt$-cut $U$ of capacity strictly less than $1$ must then contain $T \setminus C$. Therefore the set $V \setminus U$ is a strict ${\cal F}^F$-core.}
        \label{fig:auxiliary-c}
   \end{subfigure}
 \end{center}
  \caption{Constructing an Auxiliary Graph $H_t$}
  \label{fig:auxiliary}
  \hrule
\end{figure}

\begin{algorithm}[ht]
\begin{algorithmic}[1]
\REQUIRE $x = \{x_e\}_{e \in E(G^F)}$
\ENSURE either report ``$x$ is feasible'' or find a violated constraint of \eqref{eq:LPCores}.
\STATE let $E' = E(G^F) \cup E_\ell \cup F$ be a multiset of edges.
\FORALL{$t \in T$}
\IF{$C_t \in {\cal M}({\cal F}^F)$}
\FORALL{$e \in E'$}
\IF{$e \in F \cup E_\ell$ or there exists $s \in T \setminus C_t$ such that $e=rs$}
\STATE let $x^t_e \leftarrow 1$.
\ELSE
\STATE let $x^t_e \leftarrow x_e$.
\ENDIF
\ENDFOR
\STATE find a minimum $rt$-cut $U$ in $H_t=(V,E')$ with capacities $\{x'_e\}_{e \in E'}$.
\IF{the capacity of the cut is less than $\ell+1$}
\RETURN $V \setminus U$.
\ENDIF
\ENDIF
\ENDFOR
\RETURN $x$ is feasible.
\end{algorithmic}
\caption{Separation Oracle for \eqref{eq:LPCores}}
\label{alg:LPSep}
\end{algorithm}
We first show that the Algorithm~\ref{alg:LPSep} runs in time polynomial in the size of $G$. Fix some $t \in T$. We can efficiently find $C_t$, and by the characterization in Proposition~\ref{prop:minSets}, we can efficiently check if $C_t \in {\cal M}({\cal F}^F)$
. Since adjusting the graph and finding a minimum cut can be done efficiently, we get that Algorithm~\ref{alg:LPSep} runs in polynomial time.
The following two claims prove the correctness of the algorithm.
\begin{claim}
If $x$ is feasible for \eqref{eq:LPCores}, then Algorithm~\ref{alg:LPSep} verifies it.
\end{claim}
\begin{proof}
Assume $x$ is feasible for \eqref{eq:LPCores}, and let $t \in T$ be such that $C_t \in {\cal M}({\cal F}^F)$. Let $U$ be an $rt$-cut in $H_t$, and denote $X = V \setminus U$. Then $t \in X$ and $r \notin X$, and therefore $X \in {\cal F}$. If there is an edge $e \in F$ that covers $X$, then $e \in \varrho_{E'}(X)$, and the capacity of the cut is 
$\sum_{e' \in \varrho_{E'}(X)}{x^t_{e'}} \ge x^t_e + \sum_{e' \in \varrho_{E_\ell}(X)}{x^t_{e'}}= \ell+1$. If there is $s \in T \setminus C_t$ such that $s \in X$, then similarly $\sum_{e' \in \varrho_{E'}(X)}{x^t_{e'}} \ge \ell + 1$.
Otherwise, $X \in {\cal F}^F$, and moreover, $X \cap T = C_t \cap T$. Therefore $X$ is a strict ${\cal F}^F$-core, and since $x$ is feasible for \eqref{eq:LPCores}, 
$$\sum_{e' \in \varrho_{E'}(X)}{x^t_{e'}} \ge \ell + \sum_{e' \in \varrho_{E(G^F)}(X)}{x^t_{e'}} = \ell + \sum_{e' \in \varrho_{E(G^F)}(X)}{x_{e'}} \ge \ell + 1\; . $$
\end{proof}

\begin{claim}
If $x$ is infeasible for \eqref{eq:LPCores}, then Algorithm~\ref{alg:LPSep} returns a violated strict ${\cal F}^F$-core.
\end{claim}
\begin{proof}
Assume that $x$ is infeasible. Then there exists $X \in {\cal F}^F$ such that $X$ is a strict ${\cal F}^F$-core and such that $$\sum_{e \in \varrho_{E(G^F)}(X)}{x_e} < \ell + 1 \; .$$ Let $t \in X \cap T$, then $(V \setminus X,X)$ is an $rt$-cut of capacity $< 1$ in $(V,E(G^F),\{x_e\}_{e \in E(G^F)})$. Note that since $X$ is a strict core, no edge in $F$ enters $X$. Therefore $(V \setminus X,X)$ is an $rt$-cut of capacity $< \ell + 1$ in $H_t$. We conclude that the capacity of the minimum $rt$-cut $U$ found by the algorithm is strictly less than $\ell + 1$. It remains to show that $Y = V \setminus U$ is a strict ${\cal F}^F$-core. Clearly $t \in Y\cap T$ and $r \notin Y$. Let $e \in F$, then $x^t_e=1$. Since 
$$0 \le |\varrho_F(Y)| = \sum_{e \in \varrho_F(Y)}{x^t_e} \le \sum_{e \in \varrho_{E'}(Y)}{x^t_e} - \ell < 1 \;,$$
then $Y \in {\cal F}^F$. Consider $s \in Y \cap T$. If $s \notin C$, then $rs \in \varrho_{E'}(Y)$ and thus $\sum_{e \in \varrho_{E'}(Y)}{x^t_e} \ge \ell + x^t_{rs} = \ell + 1$. Therefore $C$ is the unique element of ${\cal M}({\cal F}^F)$ contained in $Y$, and $Y \cap T = C \cap T$.
\end{proof}

\begin{corollary}\label{cor:seperation}
Algorithm~\ref{alg:LPSep} is a polynomial time separation oracle for the linear program~\eqref{eq:LPCores}
\end{corollary}

\begin{proof}[Proof of Lemma~\ref{l:strictAlgorithm}]
From Corollary~\ref{cor:seperation} and Lemma~\ref{appC:hitSetEffApprox} we deduce that the randomized rounding algorithm for the hitting set problem (Algorithm~\ref{appAlg:HSRR}) outputs a set $A \subseteq E(G^F)$ such that
$w^1(A) \le O(\log|{\cal S}^F|) \opta(G,E_\ell)$. Following Lemma~\ref{l:approxValue}, $|{\cal S}^F| \le 2^{|S|}|T|$, and therefore $w^1(A) \le O(|S|) \opta(G,E_\ell)$. Moreover, Lemma~\ref{appC:hitSetEffApprox} guarantees that  with probability at least $1 - 2^{-|E|}$ the algorithm runs in polynomial time.
\end{proof}

\end{document}